\documentclass[12pt]{article}
\usepackage{amsmath}
\usepackage{graphicx,psfrag,epsf}
\usepackage{enumerate}
\usepackage[numbers]{natbib}
\usepackage{url} 

\usepackage{amsfonts}
\usepackage{amsthm}
\usepackage{bbm}

\usepackage{subfigure}
\newcommand{\blind}{0}

\usepackage{amsmath,amsfonts}

\newtheorem{prop}{Proposition}[section]
\begin{document}

\def\spacingset#1{\renewcommand{\baselinestretch}%
{#1}\small\normalsize} \spacingset{1}


\if0\blind
{
  \title{\bf A Model-based Semi-supervised Clustering Methodology}
  \author{Jordan Yoder\thanks{
    This work is partially funded by the National Security Science and Engineering Faculty Fellowship (NSSEFF),
the Johns Hopkins University Human Language Technology Center of Excellence (JHU HLT COE), and the
XDATA program of the Defense Advanced Research Projects Agency (DARPA) administered through Air
Force Research Laboratory contract FA8750-12-2-0303.}\hspace{.2cm}\\
    Department of Applied Mathematics and Statistics, Johns Hopkins University\\
    and \\
    Carey E. Priebe \\
    Department of Applied Mathematics and Statistics, Johns Hopkins University}
  \maketitle
} \fi

\if1\blind
{
  \bigskip
  \bigskip
  \bigskip
  \begin{center}
    {\LARGE\bf A Model-based Semi-supervised Clustering Methodology}
\end{center}
  \medskip
} \fi

\bigskip
\begin{abstract}
We consider an extension of model-based clustering to the semi-supervised case, where some of the data are pre-labeled.  We provide a derivation of the Bayesian Information Criterion (BIC) approximation to the Bayes factor in this setting. We then use the BIC to the select number of clusters and the variables useful for clustering.  We demonstrate the efficacy of this adaptation of the model-based clustering paradigm through two simulation examples and a fly larvae behavioral dataset in which lines of neurons are clustered into behavioral groups.
\end{abstract}

\noindent%
\vfill

\newpage
\spacingset{1.45} 
\section{Introduction}
Clustering is the art of partitioning data into distinct and scientifically relevant classes by assigning labels to observations.  Clustering is typically performed in an unsupervised setting, where none of the observed data initially have labels.  In semi-supervised learning, some of the data have labels, but others do not; the goal is to classify the unlabeled data by assigning them to the same classes as the labeled data.  In the gap between these two settings, there is semi-supervised clustering, in which some of the data have labels, and there may not be labeled data from every class.  Furthermore, the total number of classes may be unknown.  Just as in semi-supervised learning, leveraging the known labels allows for a more informed clustering.

There are many strategies for solving clustering problems.  Some clustering procedures are based on heuristics that lack a principled justification, and many of the basic questions of clustering (e.g., how many classes there are) are often left to the intuition of the practitioner.   Fraley and Raftery \cite{fraley2002model} review model-based clustering, which recasts the task of clustering as a model selection problem, which is well-studied in the field of statistics.  The main assumption in  model-based clustering is that the data are drawn from one of $G$ distributions, where each distribution represents a cluster.  Model selection can be accomplished by computing approximate Bayes factors for competing models with different numbers and/or types of components.  

Much of the recent work in model-based clustering has centered on variable selection.  Raftery and Dean \cite{raftery2006variable} proposed a greedy algorithm for model-based clustering in the unsupervised setting; it used the BIC to choose the number of clusters and the features to consider while clustering by proposing linear relationships between variables that influence clustering and those independent of clustering.  Murphy \emph{et al.} \cite{murphy2010variable} adapted Raftery and Dean's algorithm to semi-supervised learning.  Maugis \emph{et al.} \citep{maugis2009variable, maugis2009genvariable} consider many more scenarios of dependency structures between the clustering variables and the remaining variables than in \cite{raftery2006variable}.  Taking a different approach, Witten \emph{et al.} \cite{witten2010framework} offer a framework for variable selection in clustering using sparse k-means or sparse hierarchical algorithms by modifying the corresponding objective function to include penalty constraints.  In their simulations, they outperform the methods of \cite{raftery2006variable} considerably.

The BIC is generally denoted as
$$BIC = 2 \mathcal{L} - d \log(n),$$
where $\mathcal{L}$ is the maximum of the likelihood, $d$ is the number of parameters estimated, and $n$ is the number of data used to estimate those parameters.  Alternative information criteria generally differ by having different values for the penalty term.  For example, the sample size adjusted BIC \cite{boekee1981order}  is defined as $$sBIC = 2 \mathcal{L} - d \log(\frac{n+2}{24}).$$

Because $BIC - sBIC = o(1),$ one may wonder of the efficacy of such a deviation from the BIC.  This question is particularly salient when considering that the standard derivation of the BIC as the approximation to the integrated loglikelihood is only $O(1)$ accurate asymptotically (so that o(1) terms are negligible).  In this paper, we will present such a modified BIC that is asymptotically equivalent to the standard BIC and thoroughly explore an analytical example of how such a modification could be useful.  

In this paper, we quickly review model based clustering in Section \ref{review}.   We state our derivation of a modified BIC to the semi-supervised case in Section \ref{methods} that is asymptotically equivalent to the standard BIC.  We follow this with a detailed discussion of $o(1)$ equivalent information criteria, analytical calculations for a toy example, and simulations for a semi-supervised clustering example.  Next, we apply our methods to the fly larvae behavioral dataset in Section \ref{fly}.  Finally, we conclude the paper with a discussion of the results and extensions to our method.

\section{Review of Model Based Clustering}
\label{review}

Assume that the data $X_1, X_2, \dots, X_n$ are distributed i.i.d. according to a mixture model,
$$
f_\theta(\cdot) = \sum_{k = 1}^{G} \pi_k f_{\theta_k}(\cdot),
$$
where $\theta_{k}$ are the parameters of the $k^{\text{th}}$ component of the mixture and $G$ is the total number of components.
It can be shown that this is equivalent to the following generative process: for each $i \in \{1, 2, \dots, n\}$
\begin{enumerate}
\item[(1)] draw $Z_i$ from $\text{multinomial}\left(\pi_1, \pi_2, \dots, \pi_G\right)$
\item[(2)] draw $X_i$ from $f_{\theta_{Z_i}}$.
\end{enumerate}

If we consider each of the $G$ components as representing a single cluster, then the problem of clustering the data is that of unveiling each latent $Z_i$ from the generative process.  The expectation-maximization  (EM) algorithm can be used to find the maximum likelihood estimates for the parameters of the model and the posterior probabilities of cluster membership for each $X_i$ (cf. \cite{dempster1977maximum}).

With this formulation, we have a fully probabilistic clustering scheme.  We can set the cluster of the $i^{\text{th}}$ observation to the maximum a posteriori estimate:
$$\hat{Z_i} = \text{argmax}_{k} \pi_k f_k(x_i).$$  Note that the posterior probability can give us a measure of confidence about the $i^{\text{th}}$ classification.

Here, we have used $f_{\theta_k}$ as a general density function, but it should be noted that the multivariate Gaussian distribution is the most common choice of distribution.  We will eventually use an information criterion, such as the Bayesian Information Criterion (BIC), to assess the relative quality of different clusterings.  Then, the number of clusters can be chosen using the BIC, which rewards model fit and penalizes model complexity.

By considering different constraints to the covariance matrices in the Gaussian components, we can reduce the number of parameters estimated, thus lowering the influence of the penalty term in the BIC.  This allows for  simpler clusters to be chosen over complex clusters.  Celeux and Govaert \cite{celeux1995gaussian} consider a spectral decomposition of the $k^{\text{th}}$ component's covariance matrix,
$$\Sigma_k = \lambda_k D^T_k A_k D_k,$$
where $\lambda_k$ represents the largest eigenvalue, $A_k$ is a diagonal matrix whose largest entry is $1$, and $D_k$ is a matrix of the corresponding eigenvectors.  The interpretation of each term as it relates to cluster $k$ is as follows: $\lambda_k$ represents the volume, $D_k$ the orientation, and $A_k$ the shape.  By forcing one or more of these terms to be the same across all clusters, we can reduce the number of parameters to be estimated from $Gd+G\frac{d(d+1)}{2}$ in the unconstrained case ($\Sigma_k = \lambda_k D^T_k A_k D_k$) to $Gd + G - 1 + \frac{d(d+1)}{2}$ in the most constrained case ($\Sigma_k = \lambda D^T A D$).  We can also force additional constraints, such $D_k = I$ 
and/or $A_k=I$, leading to the simplest model: $\Sigma_k = \lambda I$ with only $G(d+1)$ parameters to estimate.

\section{Methodology}
\label{methods}
\label{sec:ssModel}
The main theoretical contribution of this paper is to derive an adjustment to the BIC for the semi-supervised clustering setting followed by an exploration about $o(1)$ equivalent information criteria. We will set forth a formalization of a more general setting, derive the BIC in this broader setting, and then apply our result to the special case of semi-supervised clustering.   The modified BIC then represents a principled measure by which to choose the number of clusters and the variables for clustering when performing semi-supervised clustering.

Consider $n=\sum_{i=1}^{C}n_i$ independent random variables
\begin{gather*}
X_1^{(1)}, X_2^{(1)}, \dots, X_{n_1}^{(1)} \overset{\text{i.i.d.}}{\sim} f_{\theta_1},  \\
X_1^{(2)}, X_2^{(2)}, \dots, X_{n_2}^{(2)} \overset{\text{i.i.d.}}{\sim} f_{\theta_2},\\ 
\vdots \\  
X_1^{(C)}, X_2^{(C)}, \dots, X_{n_C}^{(C)} \overset{\text{i.i.d.}}{\sim} f_{\theta_C}, 
\end{gather*}
where $\theta_1 \in \Theta_1 \subset \mathbb{R}^{d_1}$ and $\theta_j \in \Theta_j \subset \Theta_1$ for $j = 2, 3, \dots C$ and
$$ (f_{\theta_1}, f_{\theta_2}, \dots, f_{\theta_C}) \in M = \{ (f_{\theta^{(1)}}, f_{\theta^{(2)}}, \dots, f_{\theta^{(C)}}): \theta = (\theta^{(1)}, \dots, \theta^{(C)})\in \Theta_1 \times \Theta_2 \times \dots \times \Theta_C\}.$$
Collect all of the $X's$ into set $D$.

Model $M$, while more general than strictly necessary, encompasses the semi-supervised case.  Consider the first group of $n_1$ random variables as those for which we do not have labels and each of the other $C-1$ groups as those whose labels we know.  Then, for a proposed total number of clusters $G \geq C-1$, we assume
$$f_{\theta_1}(x) = \sum_{j=1}^{G} \pi_j \phi(x; \mu_j, \Sigma_j),$$ 
where $\phi(\cdot; \mu_j, \Sigma_j)$ is a multivariate normal pdf with mean $\mu_j$ and covariance matrix $\Sigma_j$, and $ \sum_{j=1}^{G} \pi_j =1$.  Also, for each $k \in \{2, 3, \dots, C\}$,
$$f_{\theta_k}(x) = \phi(x;\mu_{j_k}, \Sigma_{j_k}),$$ where the double subscript $j_k$ is to account for possible relabeling.  It follows that 
$$\theta_1 = \left( (\pi_1, \pi_2, \dots, \pi_G), (\mu_1, \mu_2, \dots, \mu_G), (\Sigma_1,\Sigma_2, \dots, \Sigma_G) \right) \in \Theta_1, $$
where
{\footnotesize
\begin{align*}
\Theta_1 = \{((\pi_1, \pi_2, \dots, \pi_G) : \sum_{i=1}^G \pi_i =1,
 \pi_i \in [0,1]\} \times \mathbb{R}^{d \times G} \times \{(\Sigma_1,\Sigma_2, \dots, \Sigma_G) | \Sigma_i \succeq 0, \Sigma_i \in \mathbb{R}^{d \times d}\}.\end{align*}
 }

Then, for each $k \in \{2, 3, \dots, C\}$, $\Theta_k$ is the same as $\Theta_1$ except that the mixing coefficients $(\pi_1, \pi_2, \dots, \pi_G)$ are constrained such that $\pi_{j_k} =1$ and all other $\pi_i = 0$.

With this model, we can derive the Expectation step (E-step) and Maximization (M-step) of the EM algorithm.  Note that semi-supervised EM is not new;  indeed, similar calculations can be found in \cite{mclachlan2004finite} Section 2.19 and \cite{shental2003computing}.  

Here, we generalize the BIC to the semi-supervised case with the aim of not overly penalizing for supervised data.  The modified BIC approximation to the integrated likelihood of model $M$ is $$2 \log(P(D|M)) \approx 2\log(P(D|\hat{\theta}, M))  - d \log (n_1) = BIC^*_M,$$ where $\hat{\theta}$ is the maximum likelihood estimate (MLE).

We provide the derivation in the appendix.  Compare the above result to the classical BIC approximation:
$$BIC \approx 2 \log(P(D|M)) \approx 2\log(P(D|\hat{\theta}, M))  - d \log (n).$$   The derivation of the BIC (adjusted and normal) involves an O(1) term, which is dropped at the end.  Due to this, the difference between the two BICs is a smaller penalty in the semi-supervised case (i.e. $d \log(n_1)$ vs. $d\log(n)$, which is a $o(1)$ difference).  This can be interpreted as not penalizing more complicated cluster structures for the supervised data.  We will  discuss the merit of such a deviation in Section \ref{sec:o1Equiv}.

It is known that finite mixture models do not meet the regularity conditions for the BIC approximation used to approximate the integrated likelihood.  However, \cite{fraley2002model} comment that it is still an appropriate and effective approximation in the paradigm of model-based clustering.  Thus, we will not be too remiss in considering Theorem \ref{BIC} as applicable when performing semi-supervised clustering.   Because the model $M$ is appropriate for the semi-supervised case, we can use the modified BIC as a way to calculate Bayes factors for different number of total clusters $G$, subsets of variables to be included in the clustering, and parameterizations of $\Sigma_i$.  The largest BIC value will therefore correspond to our chosen model for clustering.  We will call the semi-supervised clustering algorithm using the BIC from Theorem \ref{BIC} to make the aforementioned selections \emph{ssClust} henceforth.

\section{o(1) equivalent Information Criteria}\label{sec:o1Equiv}

Here, we comment on the use of different penalty terms.  Suppose that we define an adjusted BIC, say $BIC'$ as a function of $m \in [1, \inf)$ as $BIC'(m) = 2 \mathcal{L} - d \log(m) = BIC - d\log(\frac{m}{n}).$

When does using such an adjustment matter?  Consider two models, say $M_0$ and $M_1$, in competition for selection.  Fix an $m$.  Suppose that
\begin{enumerate}
\item[(i)] $d_1 > d_0$ 
\item[(ii)] $BIC_0 > BIC_1$ and
\item[(iii)] $BIC'_0(m) < BIC'_1(m).$
\end{enumerate}
In this case, we see that under the original definition of the BIC, we would choose $M_1$ over $M_0$, and with the modified definition, we would do the opposite.

In order to analyze this under some assumptions, define two statistical tests:
$$T(\omega) = \mathbbm{1}\{BIC_1 > BIC_0\}$$ and $$T'_m(\omega) = \mathbbm{1}\{BIC'_1(m) > BIC'_0(m)\}.$$

When do $T$ and $T'$ differ?  

\begin{enumerate}
\item[Case 1:] $T=1$ and $T'_m = 0.$  For $m \leq n$ this is impossible.

\item[Case 2:]  $T = 0$ and $T'_m = 1.$  We have two subcases.
\begin{enumerate}
\item[(a)] Suppose $M_1$ is the true model.  We would be correct in choosing $T'_m$ over $T$. 
 \item[(b)]  Suppose $M_0$ is the true model.  Note that in this case, we would make a mistake by choosing $T'_m$ over $T$.
\end{enumerate}
\end{enumerate}

We would like to analyze the probabilities in case 2.  Preferably, (2.a) is much more probable than (2.b).  Unfortunately, under (2.a) the distribution of $W_{1,0}$ is only known for certain cases.  For example, with local alternatives of the form $\theta_n = \theta_0 + \frac{\Delta}{\sqrt{n}},$  $W_{1,0}(d_1 - d_0)$ is known to have noncentral chi square distribution with $df= d_1 - d_0$ and the appropriate non-centrality parameter.  Since this result requires rather quixotic assumptions to hold, we will not theoretically bound the probabilities for more general cases; instead, we will defer our analysis to a specific example.

We can say more about (2.b) with some assumptions.
\begin{prop}
Suppose the models are nested (so that $M_0$ is a submodel of $M_1$).
Define $W_{1,0} := \frac{2( \mathcal{L}_1 - \mathcal{L}_0)}{(d_1 - d_0)}.$
Under $M_0$, $T = 0$ and $T'_m = 1$ with probability 
\begin{equation} 
Pr_{M_0} \left(\log(m) < W_{1,0} < \log(n)\right) \rightarrow F(\log(n)) - F(\log(m)),
\label{eqn:badProb}
\end{equation}
where $F(\cdot)$ is the distribution function of an $F$ distribution with $df=(d_1 - d_0, \infty).$ 
\end{prop}

\begin{proof}
We can do some algebra on $(ii)$ and $(iii)$ to obtain equivalent condition that 
$$ (d_1 - d_0) \log(m) < 2( \mathcal{L}_1 - \mathcal{L}_0) < (d_1 - d_0) \log(n).$$

If the models are nested (so that $M_0$ is a submodel of $M_1$), then the numerator of $W_{1,0}$ converges in law to a $\chi^2$ distribution with $df = d_1 - d_0$ by Wilk's Theorem.  In this case, we note that by Slutsky's Theorem, $W_{1,0}$ asymptotically has the distribution of an F-statistic with $df=(d_1 - d_0, \infty).$

Hence, under the nested models assumption, the probability that $(ii)$ and $(iii)$ hold when they should {\bf not} (i.e. choose $M_1$ when $M_0$ is the truth) is given by equation (\ref{eqn:badProb}).
\end{proof}

\subsection{Illustrative Example}
We now present an analysis for specific null and alternative models where we can explicitly calculate the probabilities of cases (2.a) and (2.b).  The purpose is to demonstrate that for finite $n$, the choice of penalty term is nontrivial.

Let $d, d_0 \in \mathbb{N}$ with $d_0 < d$.
Consider the nested models:
\begin{itemize}
\item $M_1 = \{ N(\mu, I) : \mu \in \mathbb{R}^d\}$
\item $M_{0, d_0} = \{ N(\mu'_{d_0}, I) : \mu' \in \mathbb{R}^{d}, \mu'_j = 0 \mbox{ for } j = d_0+1, d_0+2, \dots, d\},$
\end{itemize}
where $I$ is the $d-$dimensional identity matrix.

Let
\begin{gather*}
    \mu^* = \left[\begin{matrix}  1 & \frac{1}{\sqrt{2}} & \frac{1}{\sqrt{3}} & \dots & \frac{1}{\sqrt{d}} \end{matrix} \right]^T \\
     \mu^*_{0, d_0} = \left[ \begin{matrix}  1 & \frac{1}{\sqrt{2}} & \frac{1}{\sqrt{3}} & \dots & \frac{1}{\sqrt{d_0}} & 0 & 0 & \dots & 0\end{matrix}\right]^T.
\end{gather*}

Then, $f_1 := N(\mu^*, I) \in M_1$ and $f_{0, d_0} := N(\mu^*_{0,d_0}, I) \in M_0.$
Suppose that $X_1, X_2, \dots, X_n \overset{i.i.d.}{\sim} f_1 \mbox{ or } f_{0, d_0}$ according to whether or not $M_1$ is the true model.

Fix a realization of the data $(x_1, x_2, \dots, x_n).$  Let $\bar{x} = \frac{1}{n}\sum_{i=1}^n x_i$ and $\bar{x}_{0, d_0}$ be $\bar{x}$ with the coordinates after $d_0$ set to 0.
Twice the maximized loglikelihood of the data under $M_1$ is (up to constants) $$2\mathcal{L}_1 := -\sum_{i=1}^n (x_i - \bar{x})^T(x_i - \bar{x}).$$
Under $M_{0,d_0}$ it is 
\begin{align*}
2\mathcal{L}_{0, d_0} :=& -\sum_{i=1}^n (x_i - \bar{x}_{0, d_0})^T(x_i - \bar{x}_{0, d_0}) \\
=&  \mathcal{L}_1 - n \|\bar{x} - \bar{x}_{0, d_0}\|_2^2, 
\end{align*}
since $2 \sum_{i=1}^n (x_i - \bar{x})^T( \bar{x} - \bar{x}_{0, d_0})=0$. Thus we find that $\mathcal{L}_{0, d_0} < \mathcal{L}_1.$
\begin{enumerate}
\item[2a)] Under $f_1$,  \begin{eqnarray*}
n \|\bar{x} - \bar{x}_{0, d_0}\|_2^2 &=& \sum_{j=d_0+1}^d (\sqrt{n}\bar{x}_j)^2  := Y_{1} , 
\end{eqnarray*}
where $Y_{1} \sim$ noncentral $\chi^2_{df = (d - d_0)}(n\sum_{j=d_0+1}^d \frac{1}{j}).$   Hence, for any $m> 0,$
\begin{align*}
Pr_{f_1}\Big( 2 \mathcal{L}_1 - d \log(m)& < 2 \mathcal{L}_{0, d_0} - d_0 \log(m) \Big)\\
& = Pr\left(  Y_{1} < (d - d_0) \log(m)\right) .
\end{align*}
Therefore, 
\begin{align*}
Pr_{f_1}&(T = 0\  \ \text{and}\  \ T'_m = 1) =\\
 &Pr\left( (d - d_0) \log(m) < Y_{1} < (d - d_0) \log(n) \right).
\end{align*}
\item[2b)] Under $f_{0, d_0},$
\begin{equation}
n \|\bar{x} - \bar{x}_{0, d_0}\|_2^2 = Y_{d_0},
\end{equation}
where $Y_{d_0} \sim \chi^2_{df = (d - d_0)}$.  Thus,
\begin{align*}
Pr_{f_{0,d_0}}&\left( 2 \mathcal{L}_1 - d \log(m) < 2 \mathcal{L}_{0, d_0} - d_0 \log(m) \right) \\
&= Pr\left( Y_{d_0} < (d - d_0) \log(m) \right).
\end{align*}
Hence, 
\begin{align*}
Pr_{f_{0,d_0}}&(T = 0 \mbox{ and } T'_m =1)=\\
& Pr\left(  (d - d_0) \log(m) < Y_{d_0} < (d - d_0) \log(n) \right).
\end{align*}

\end{enumerate}

Clearly, we can vary $n, m, d,$ and $d_0$ to influence these probabilities.  We will now do so to show how (a) some penalty is better than none; (b) the AIC penalty can result in mistakes; and (c) the optimal penalty depends on $n, m, d,$ and $d_0$.

Figure \ref{fig:probsNVary} depicts the probability of (2.a) and (2.b) when $n$ and $m$ vary for fixed $d=200$ and $d_0 = 190$ under the alternative (a) and null (b) hypotheses.  Lower penalties than the BIC would use generally result in better decisions under $f_1$.  However, there is approximately a $3\%$ chance of error under $f_0$ when the BIC would be correct with the penalty of $\exp(2),$ that of the AIC.

\begin{figure}[t!]
    \centering
\includegraphics[width=.95\textwidth]{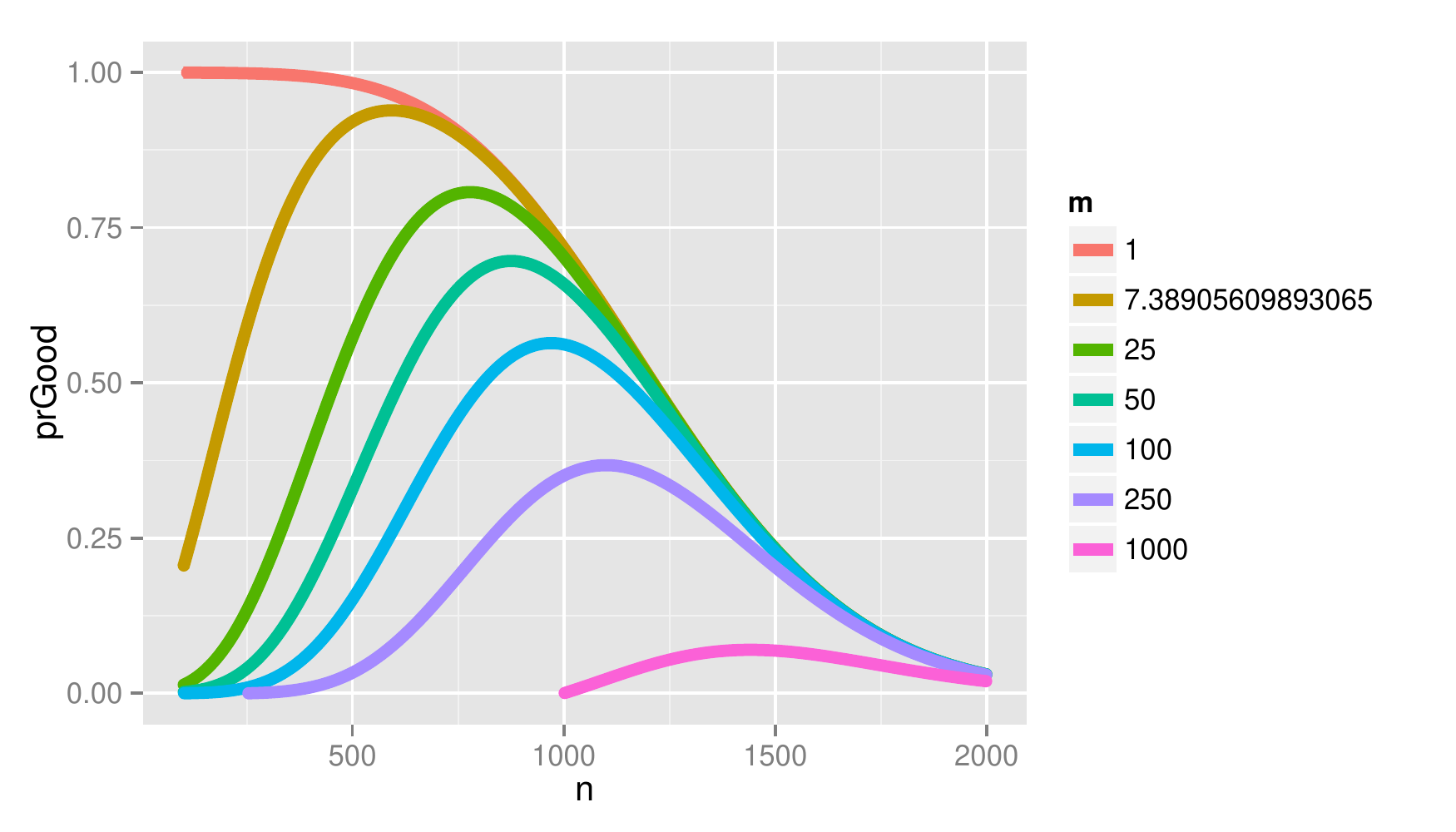}
      \caption{Probabilities of correctly choosing the alternative over the null when the BIC would not do so.  For all curves, $d = 200$ and $d_0 = 190.$  Note that $7.389... = \exp(2),$ the AIC penalty.}
    \label{fig:probsNVary}
\end{figure}

Figure \ref{fig:probsd1Vary} depicts the probability of (2.a) and (2.b) when $d$ and $m$ vary for fixed $n=1000$ and $d_0 = d - 10$ under the alternative (a) and null (b) hypotheses.  Lower penalties than the BIC would use generally result in better decisions under $f_1$ without sacrificing making too many new mistakes under $f_0$.

\begin{figure}[t!]
    \centering
\includegraphics[width=.95\textwidth]{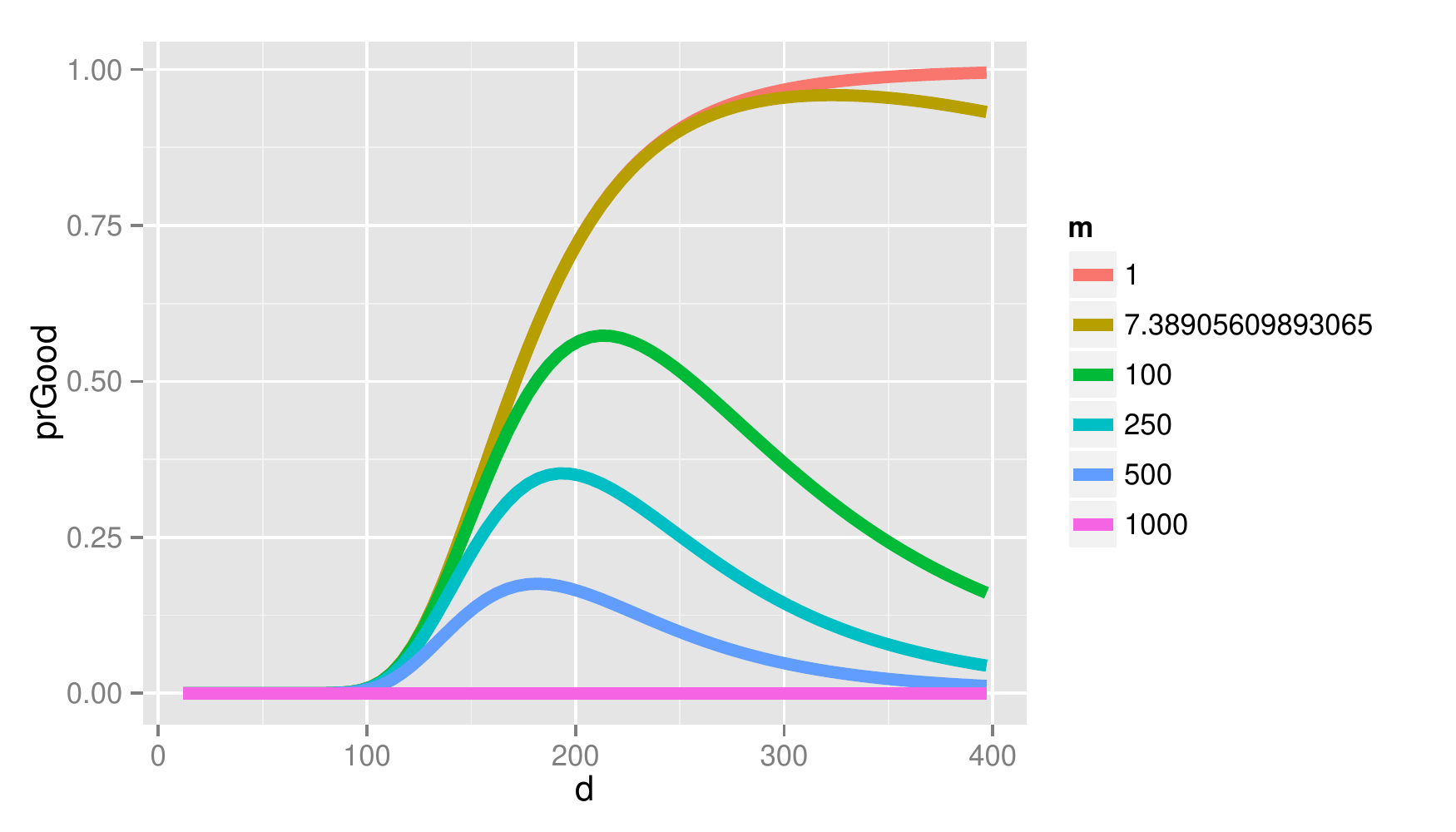}
      \caption{Probabilities of correctly choosing the alternative over the null when the BIC would not do so.  For all curves, $n = 1000$ and $d_0 = d-10.$}
    \label{fig:probsd1Vary}
\end{figure}

Figure \ref{fig:probsGapVary} depicts the probability of (2.a) and (2.b) when $d_0$ and $m$ vary for fixed $n=1000$ and $d= 200$ under the alternative (a) and null (b) hypotheses.  Lower penalties than the BIC would use generally result in better decisions under $f_1$ and not too much worse decisions under $f_0$.  However, there is approximately a larger chance of error under $f_0$ when the BIC would be correct with the AIC penalty for smaller $d_0$.

\begin{figure}[t!]
    \centering
\includegraphics[width=.95\textwidth]{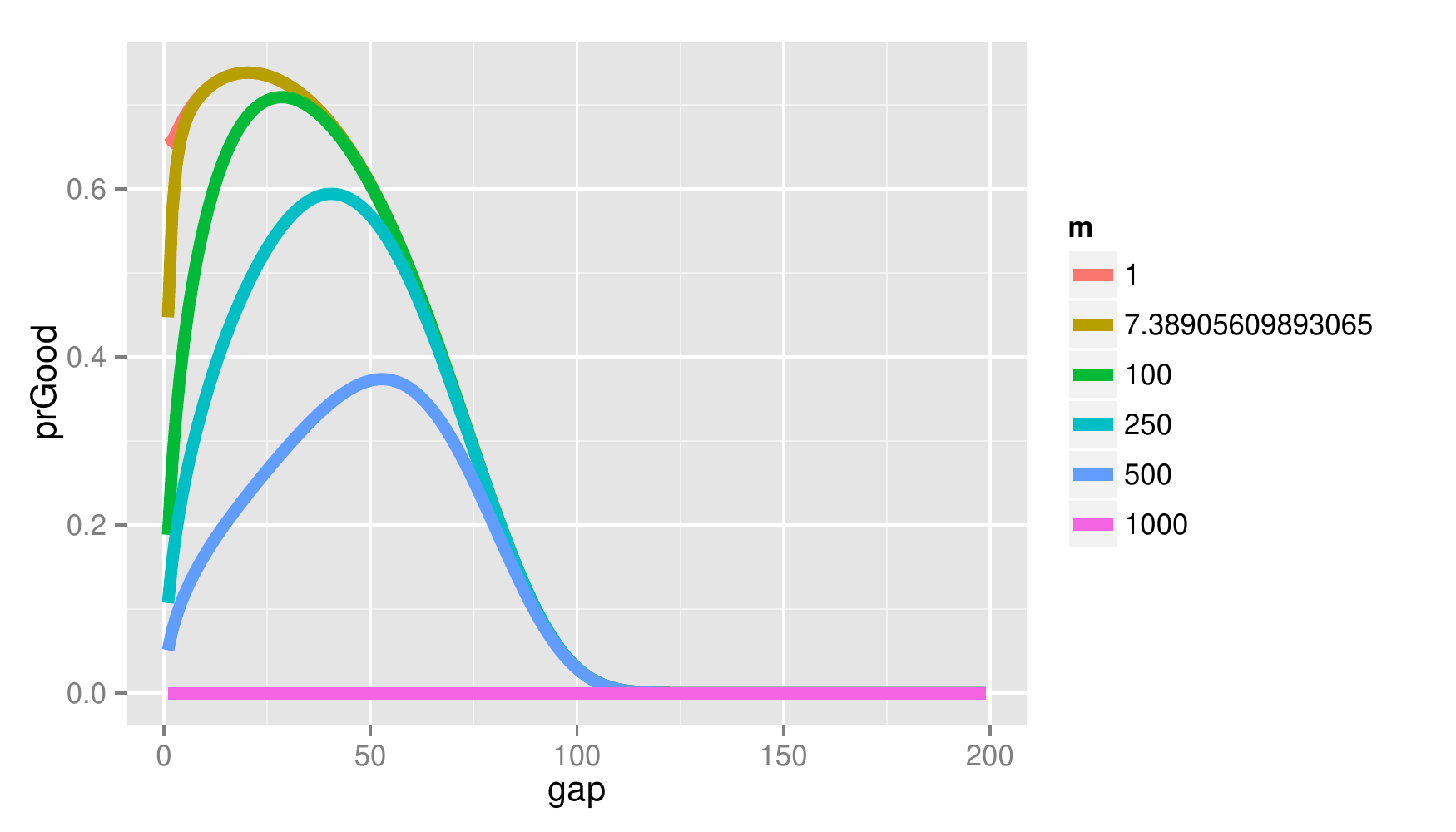}
      \caption{Probabilities of correctly choosing the alternative over the null when the BIC would not do so.  For all curves, $n = 1000$ and $d=200.$}
    \label{fig:probsGapVary}
\end{figure}

This example may seem overly simplistic.  Indeed, it does not involve semi-supervised clustering at all.  However, it does demonstrate the complexities of penalization in model selection with finite samples.

\subsection{Simulation Study}

We would like to demonstrate the impact of the previously discussed model selection complexities from the illustrative example on the inference task of semi-supervised clustering.  To that end, consider the following procedure for constructing a dataset $\mathcal{X} \subset \mathbb{R}^2$, semi-supervising it, and clustering it using model-based clustering with different penalties.  We will then compare the resulting ARI scores.

   Let $$\mu_1 =  \mu_2 = 
  \begin{bmatrix}
  0     \\
   0   
\end{bmatrix} \mbox{ and } \mu_3 =   \begin{bmatrix}
  2   \\
   2   
\end{bmatrix}.$$ Also, let $$\Sigma_1 = \begin{bmatrix}
    .5  & .35 \\
    .35 & .5 
\end{bmatrix}$$ and $$\Sigma_2 = \begin{bmatrix}
    .5    &  -.35 \\
    -.35 & .5 
\end{bmatrix}.$$  Consider the following procedure for generating one Monte Carlo Replicate. 

\begin{enumerate}
\item Define the pdf of each datum as being from a mixture model $f(X) = \sum_{k=1}^{3} \pi_k f_k(X),$ where $f_k=N(\mu_k, \Sigma_k)$ with $6\Sigma_3$ generated using the onion method of \citep{joe2006generating}
\item Draw $n^S = 100$ supervised from the first two components.
\item Draw $n^U = 5, 10, 20, 40, 80, \dots, 640$ unsupervised from the full mixture model.
\item Cluster using $2 - 5$ Gaussians with various constraints on the proposed covariance matrices.
\item Choose the models for clustering based on different values of $m$ in the penalty term $d \log(m),$ where $m \in [n^U, n^U + n^S].$  
\item Calculate resulting ARIs using the chosen models..
\end{enumerate}

Figures \ref{fig:bicPenalty} and \ref{fig:bicFinalPenalty} shows how different penalties can affect the overall clustering performance on the unlabeled data.  Generally, we see that with more unsupervised data, there are less differences between the different choices of penalty terms in the interval $[n^U, n^U+n^S],$ which makes sense given the asymptotic results.  In this particular example, less penalization allowed us to detect the third component sooner, improving the ARI values in the resulting clustering.  Thus, we have shown that in an example closer to our problem of interest, the restricted penalization derived in the previous section can be efficacious as compared to the standard penalty.  Additionally, as $n$ grows, the differences between the information criteria that are dependent on $n$ to the same order becomes negligible even for relatively small values of $n$.

\begin{figure}[t!]
    \centering
    \subfigure[$n=105$] {\includegraphics[width=.45\textwidth]{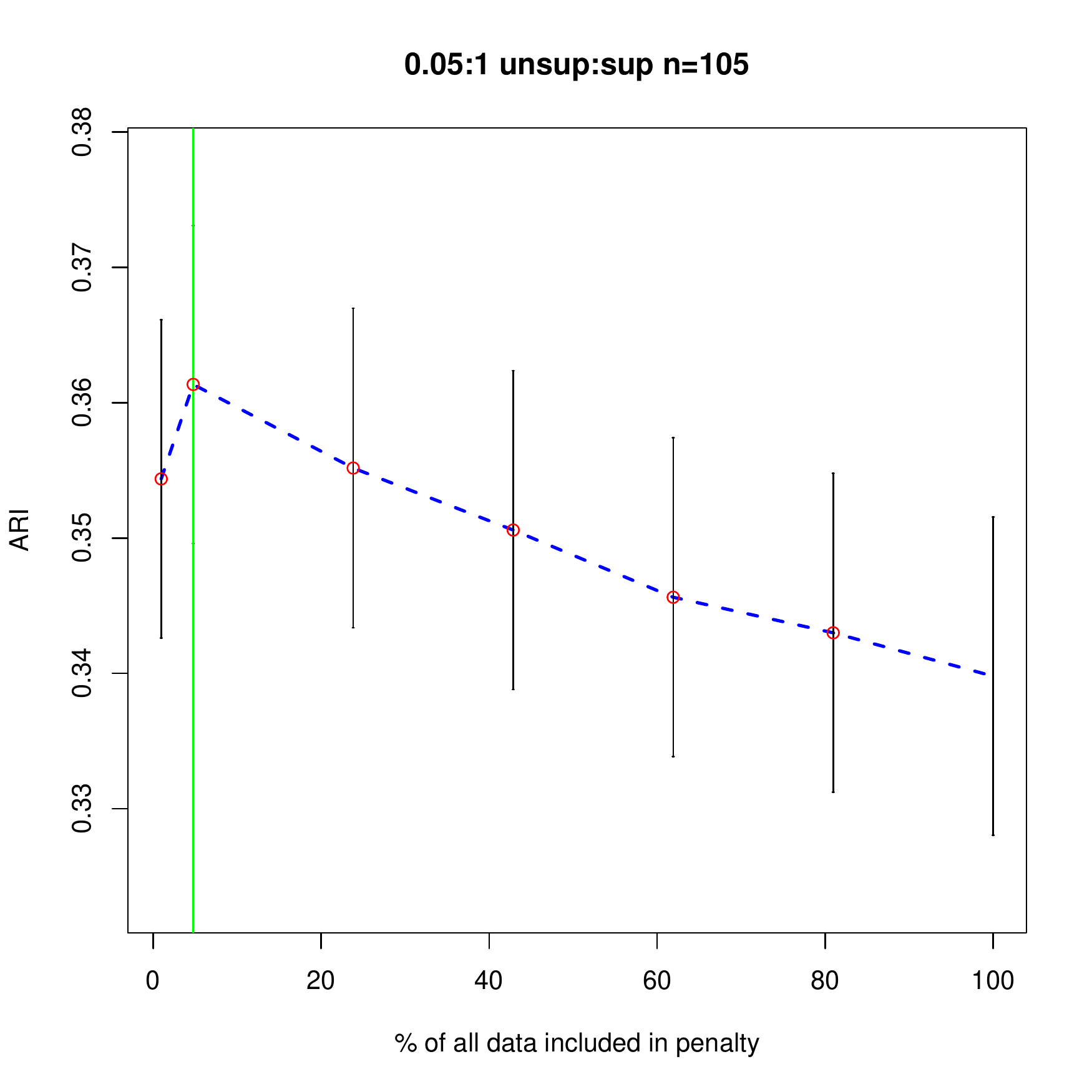}}
      \subfigure[$n=110$] {\includegraphics[width=.45\textwidth]{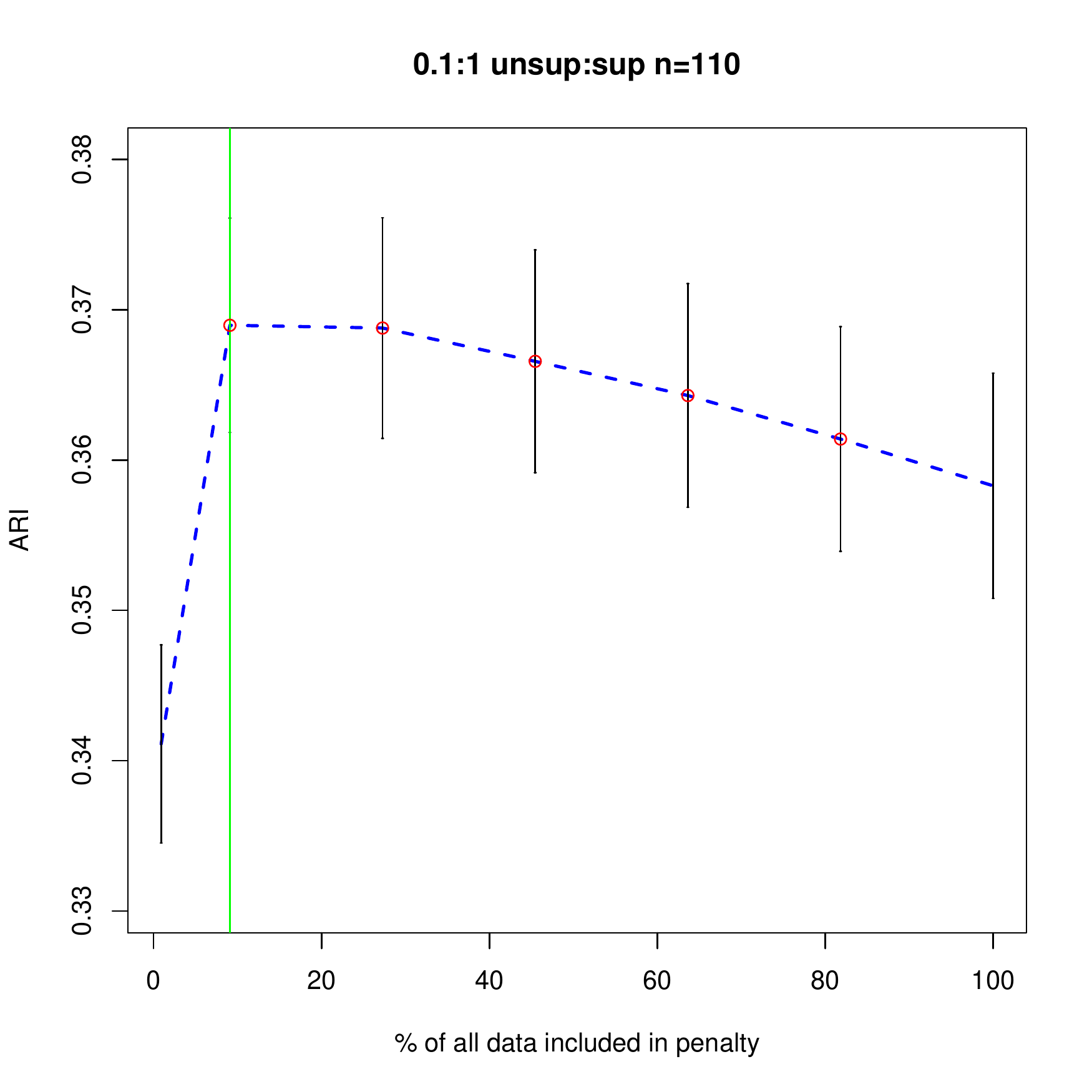}}
       \subfigure[$n=125$] {\includegraphics[width=.45\textwidth]{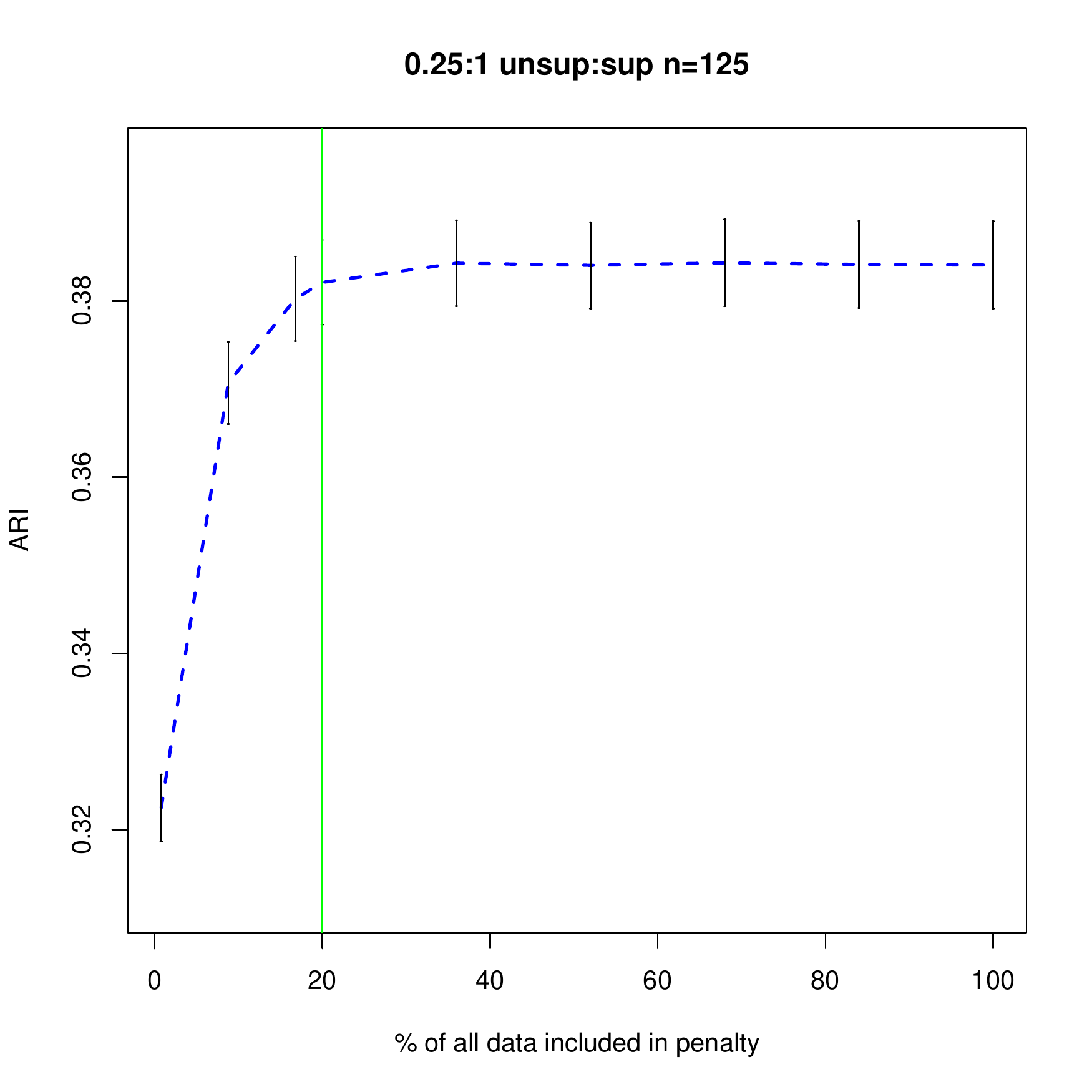}}
        \subfigure[$n=150$] {\includegraphics[width=.45\textwidth]{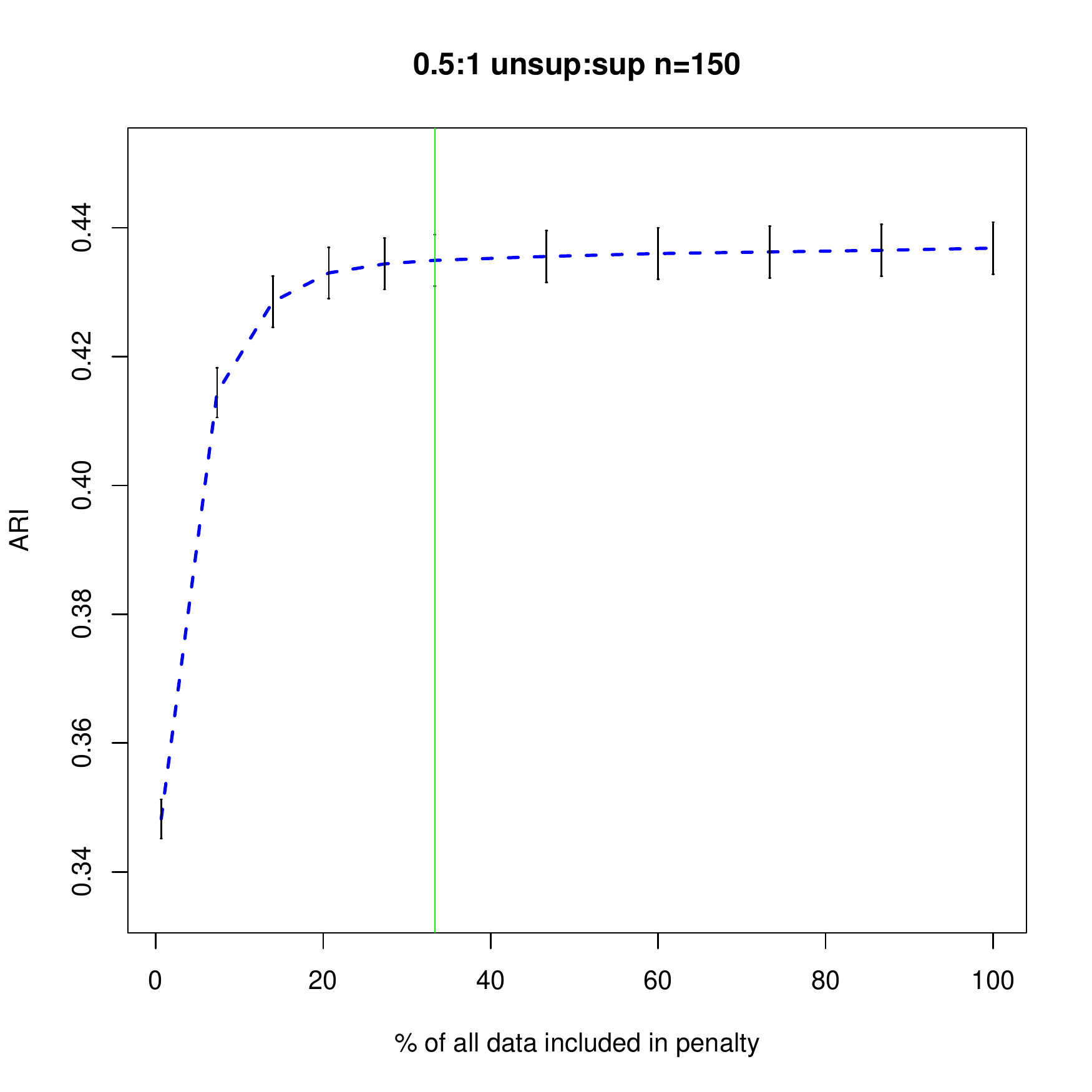}}
      \caption{ARI vs penalty in $BIC'(m).$  Higher is better.  Red circles represent significant paired Wilcoxon tests for larger ARI values than vs the standard BIC (at the .05 level).  The area to the right of the green line represents the interval $[n^U, n^U + n^S].$ }
    \label{fig:bicPenalty}
\end{figure}

\begin{figure}[t!]
    \centering
    \includegraphics[width=.9\textwidth]{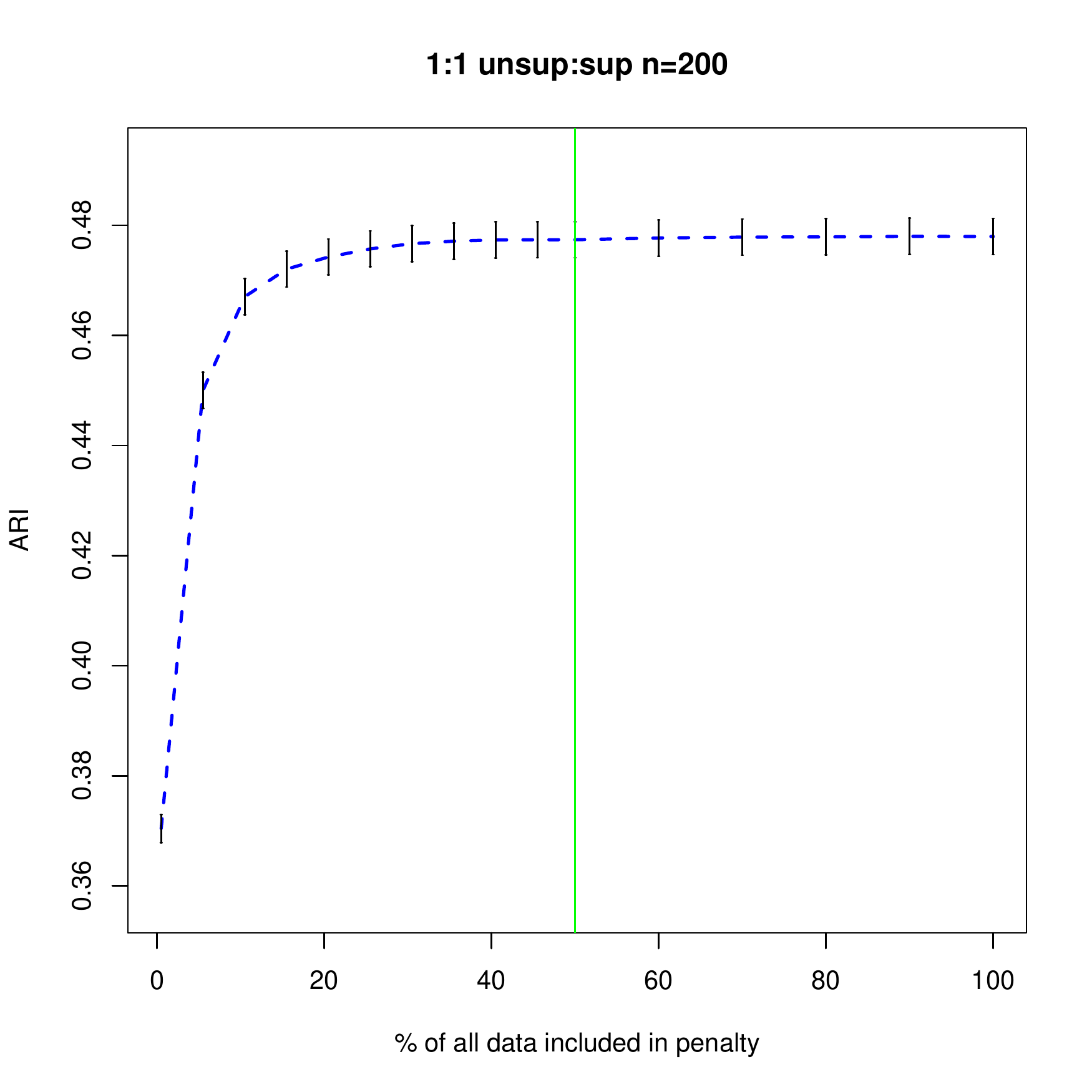}    
    \caption{ARI vs penalty in $BIC'(m)$ for largest value of the total dataset simulated.  Higher is better.  Red circles represent significant paired Wilcoxon tests for larger ARI values than vs the standard BIC (at the .05 level).  The area to the right of the green line represents the interval $[n^U, n^U + n^S].$}
    \label{fig:bicFinalPenalty}
\end{figure}

\section{Identifying Fly Behaviotypes}\label{sec:flyExp}
In one of the motivating applications for this work, classes of neurons in \emph{Drosopholia} larvae are controlled using optogenetics (cf. \cite{optoMOTY} regarding optogenetics).  In \cite{vogelstein2014discovery}, they observe the reactions of the affected larvae to stimuli in high-throughput behavioral assays. The goal is to determine which classes of neurons cause similar changes in behavior when deactivated.

We initially collected data on $n = 37780$ larvae grouped into $b = 2062$ dishes. By changing the optogenetic procedure, $\ell = 11$ known lines are created, pbd1, 38a1, 61d0, ppk1, 11f0, pbd2, 38a2, pbd3, ppk2, iav1, and 20c0.  Of these lines, we discarded the larvae in pbd3 and ppk2 because they had less than 40 larvae each.  Further, we discarded all larvae with an unknown line.  After curating the data, we now have $n=7730$ larvae.  Each larva is observed while responding to various stimuli.  Vogelstein \emph{et al.} cite{vogelstein2014discovery}  expand on the methods of \cite{priebe2001olfactory} and \cite{priebe2004integrated} and describe how the observations are embedded into $\mathbb{R}^d$, where here $d=30$.  We use the method presented in \cite{zhu2006automatic} to select the elbow of the scree plot to further reduce the data to $d=14$ dimensions.  We did not perform any additional feature selection.

For each Monte Carlo replicate, we use a small subset of the data where the line was known (101 randomly chosen animals from each of the 9 remaining lines) along with $m=0, 1, 2, \dots 8$ pre-labeled data randomly chosen from the 101 animals in each line.  We wrote an R package entitled {\tt ssClust} to perform semi-supervised GMM with similar options to the popular Mclust software, which is an R package for GMM \citep{mclustVersion4}.  We cluster the points using both {\tt ssClust} and Mclust.  Then, we compute the ARI against the line type for both methods.  

We observe that the initialization strategy of using {\tt ss-k-means++} instead of hierarchical clustering results in a significant improvement to the ARI even with no supervision. We find that a single animal per line significantly improves the clustering results, as expected (cf. Figure \ref{fig:martaExper}).  Further, we can see that there are diminishing returns on additional supervision starting at $3$ supervised examples per line. 

\begin{figure}[!ht]
   \centering      
   \includegraphics[ scale=0.5]{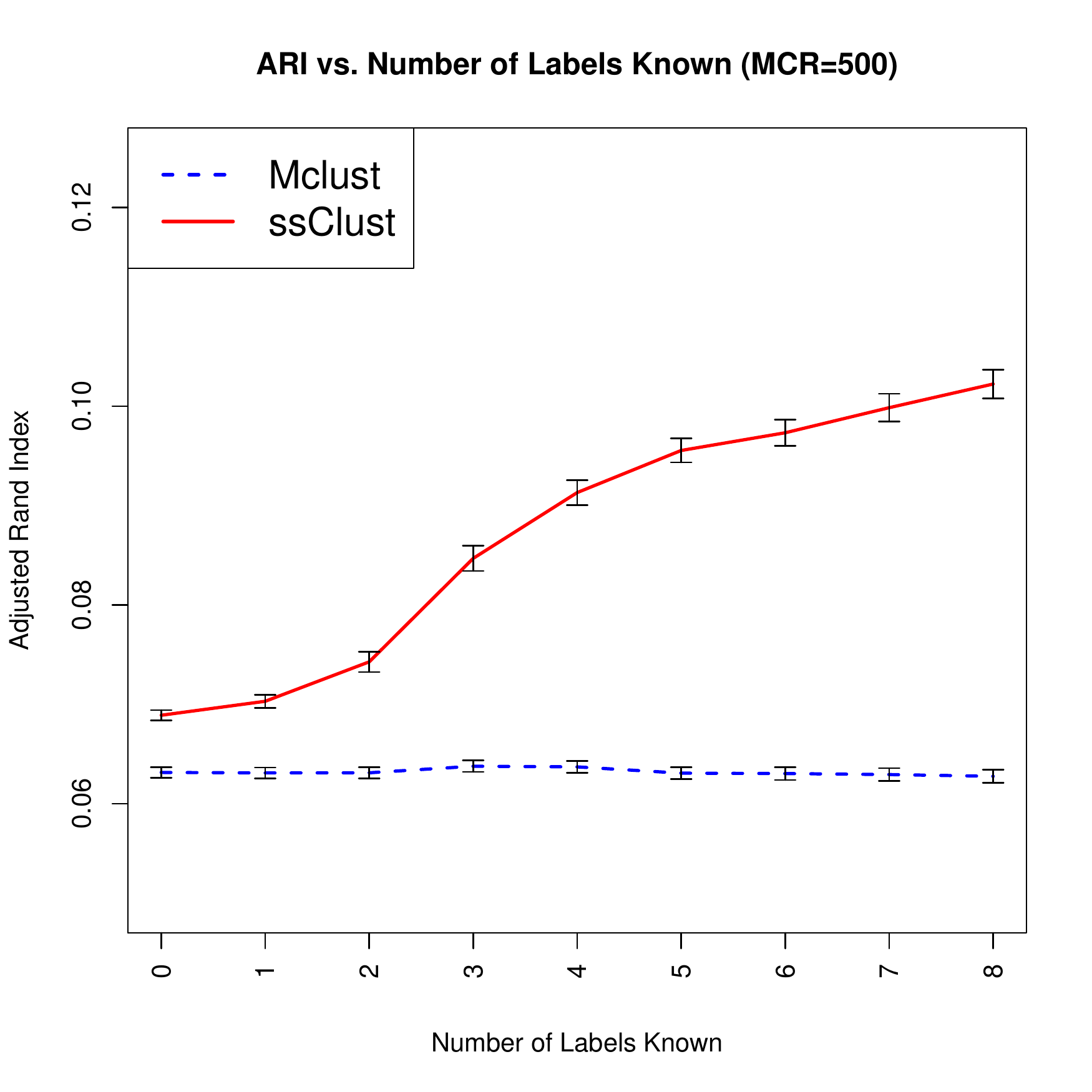}      
 \caption{Average value of the ARI for 500 Monte Carlo replicates for experiment \ref{sec:flyExp}.    Error bars are $\pm 2$ standard errors.} \label{fig:martaExper}
\end{figure}

\subsection{Differentiating Lines} \label{exp4}
Vogelstein \emph{et al.} \cite{vogelstein2014discovery} posed and answered questions of the form, ``Is line X different than line Y (in terms of behaviotypes)?''  As an illustrative example, we will perform a similar analysis comparing the ppk1 and pbd1 lines but will additionally incorporate some supervision.  Here, our interpretation of the clusters will shift from lines to behaviotypes. Our proposed procedure for distinguishing between lines is as follows:
\begin{enumerate}
\item[(1)]  Sample $101$ animals from each line
\item[(2)]  Label $3$ animals from each line according to a labeling strategy (see below)
\item[(3)]  Cluster the animals using {\tt ssClust} and Mclust into between 2 and 12 clusters using the parameterizations EEE, VVV, VII, and EII.
\item[(4)]  Collect the results and construct empirical probability of cluster membership for each line
\item[(5)]  Compute the Hellinger distance between the two lines to be compared and store this as statistic $H$
\item[(6)]  Simulate the distribution of $H$ under the null hypothesis that the lines are the same by permuting the labels and computing the Hellinger distance $H_i$ for $i = 1, 2, \dots, B$ for some large integer $B$.
\item[(7)]  Return an empirical p-value based on steps (5) and (6).
\end{enumerate}

In item (2) we did not specify a labeling strategy in detail.  We propose a strategy that is reasonably realistic to execute for our particular dataset.  Vogelstein \emph{et al.} \cite{vogelstein2014discovery} used a hierarchical clustering scheme in which the first few layers were visually identifiable by watching the worms.  Thus, by using their labels from an early layer (layer 2, with 4 clusters total), we have a plausible level of supervision for a human to have performed.  Specifically, we sample at random a label from the true labels among a line with weights proportional to the counts of each label in that line.  Using that label, we sample 3 worms of that label in that line to be the supervised examples.  Next, for the other line, we sample a different label, and 3 examples with that true label.

We see that based on all three labeling strategies that both {\tt ssClust} and Mclust are able to corroborate the results from \cite{vogelstein2014discovery} even with small amounts of data:
ppk1 and  pbd1 are statistically different (p-value $\approx 1e-4$ for both for 500 MC replicates).

We now show that {\tt ssClust} can answer these questions ``sooner" than Mclust.  That is, the p-value ($pVal$) will be below the significance level with fewer unsupervised examples.  To quantify this concept, we introduce the ``answering time" for algorithm A:

$$\tau_A :=\min_{n \in \mathbb{N}}\Big\{n: 2 \leq n \leq 30 \mbox{ and }  \prod_{q = n}^{30} \mathbbm{1}\{pVal \left(\mathcal{D}_k\right) \leq \alpha\} = 1 \Big\},$$

where here we use the notation $$\mathcal{D}_k := \{X_{1,k}, \dots, X_{q,k}, (X_{99,k}, Y_{99,k}), \dots, (X_{101,k}, Y_{101,k}) : k \in \{1, 2\}\}$$ to be the labeled and unlabeled data.

The p-value constraint bears some explanation; it says that we require a significant p-value for all datasets at least as large as with $n$ unsupervised examples per line.  Note that here we assume our datasets are nested and that they all use the same supervised examples.  Since the answering time will be dependent on the datasets used, we perform a Monte Carlo simulation with 500 random sequences of datasets and report the answering times for both {\tt ssClust} and Mclust (cf. Figure \ref{fig:answeringTime}).  The median answering time for {\tt ssClust} is significantly lower than for Mclust (p-value = 4.7e-12 for paired Wilcoxon signed-rank test).

\begin{figure}[!ht]
   \centering      
   \includegraphics[ scale=0.5]{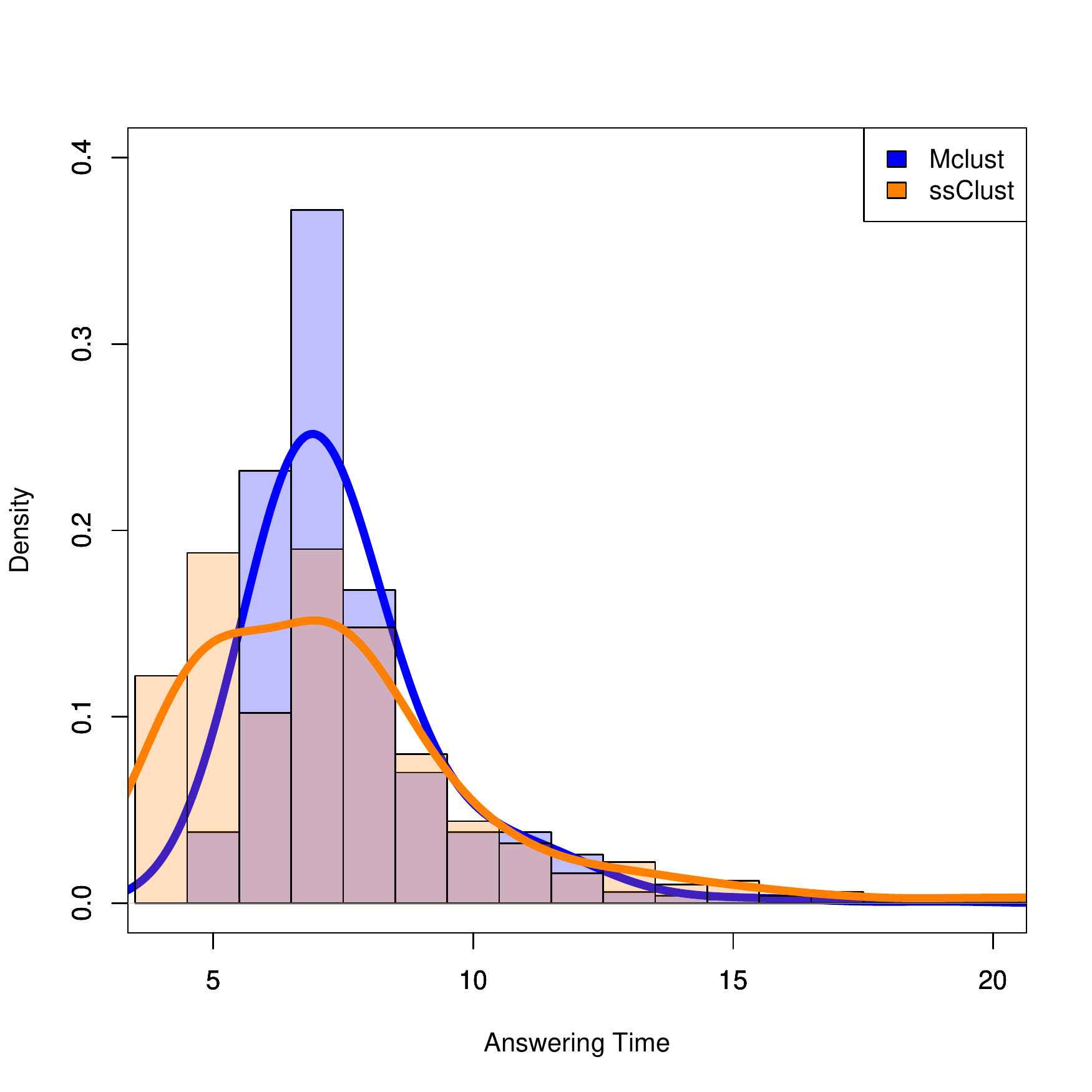}      
 \caption{Frequency distribution of answering times for \ref{exp4} given 500 MC replicates. }
 \label{fig:answeringTime}
\end{figure}
\section{Conclusion}
Previously explored approaches to semi-supervised clustering include a modified K-means, which can be seen as our algorithm constrained to spherical and identical covariance matrices without the adjusted BIC for model selection  (cf. \cite{basu2002semi}).  Others have used latent variables representing clusters and cluster-dependent distributions to give a probabilistic approach to a slightly different problem, where instead of labels being known, only pairwise constraints in the form of must-link and cannot-link are given  (cf. \cite{basu2004probabilistic}).  \cite{wagstaff2001constrained} applied an appropriately modified K-means to this problem.  Finally,  \cite{cohn2003semi} and \cite{xing2002distance} offer a different approach to semi-supervised clustering involves training a measure to account for constraints or labels.

The main contribution of this paper over previous works is presenting a probabilistic framework for semi-supervised clustering and deriving an information criterion consistent with the framework to allow selection of number of clusters and/or clustering variables.  With the corrected BIC, we found that in the simulated examples and fly larvae dataset our method outperformed the most comparable method, Mclust.  In the fly dataset, ssClust was able to recover lines better and yielded a lower answering time more often for the question of whether two particular lines were different, behaviorally.  This indicates that incorporating even a small amount of information can help guide the clustering process.

Future areas of research will involve handling the more nuanced must-link and cannot-link constraints, which are flexible enough to encompass the labels-known problem we have explored in this paper.

\section{Appendix} \label{sec:BICderiv}
{\bf A Derivation of the BIC for the Semi-supervised Model}\\
Assume the data are distributed according to a member of model $M$ described in Section \ref{sec:ssModel}.

 Consider the integrated likelihood: \begin{equation} \label{eqn:intll}
 P(D) = \int P(D|\theta, M) P(\theta| M) d\theta,
 \end{equation} where $P(\cdot)$ is a probability, pmf, or pdf where appropriate.  Let $$Q(\theta) =\log\left(P(D|\theta, M) P(\theta| M)\right),$$ the log of the posterior likelihood.  Suppose that the posterior mode exists, say $\bar{\theta}$.

A second order Taylor expansion about $\bar{\theta}$ gives
\begin{eqnarray*}
Q(\theta) &=& Q(\bar{\theta}) + (\theta - \bar{\theta})^T \nabla Q(\bar{\theta}) + \frac{1}{2}  (\theta - \bar{\theta})^T \nabla^2 Q(\bar{\theta}) (\theta - \bar{\theta}) + o(\|(\theta - \bar{\theta})\|_2^2).
\end{eqnarray*} 
 By the first order optimality necessary conditions, we know  $\nabla Q(\bar{\theta}) =0.$
Note $o(\|(\theta - \bar{\theta})\|_2^2) \rightarrow 0$ faster than $\|(\theta - \bar{\theta})\|_2^2$ as $\theta \rightarrow \bar{\theta}$.  By ignoring the last term, we will approximate $Q(\theta)$ with the truncated Taylor expansion:
\begin{eqnarray*}
Q(\theta) \approx Q(\bar{\theta}) + \frac{1}{2}  (\theta - \bar{\theta})^T \nabla^2 Q(\bar{\theta}) (\theta - \bar{\theta}).
\end{eqnarray*} 

Recalling (\ref{eqn:intll}), we may approximate $P(D|M)$ using a saddle point approximation:
\begin{eqnarray*}
P(D| M) &=&\int \exp\left( \log ( P(D|\theta, M) P(\theta| M) ) \right)d\theta \\
&=& \int \exp (Q(\theta ))d\theta \\
&\approx& \int \exp \left(Q(\bar{\theta})+ \frac{1}{2}  (\theta - \bar{\theta})^T \nabla^2 Q(\bar{\theta}) (\theta - \bar{\theta})\right)d\theta \\
&=&   \exp \left(Q(\bar{\theta}) \right) \int \exp \left( \frac{1}{2}  (\theta - \bar{\theta})^T \nabla^2 Q(\bar{\theta}) (\theta - \bar{\theta})\right)d\theta \\
&=&   \exp \left(Q(\bar{\theta}) \right) \int \exp \left( -\frac{1}{2}  (\theta - \bar{\theta})^T (- \nabla^2 Q(\bar{\theta})) (\theta - \bar{\theta})\right)d\theta. \\
\end{eqnarray*} 
Recognize the integral as proportional to the density of a multivariate Guassian with mean $\bar{\theta}$ and covariance $- \nabla^2 Q(\bar{\theta})$.  Let $H = - \nabla^2 Q(\bar{\theta}).$
Then, we have 
\begin{eqnarray}\label{eqn:laplaceapx}
P(D|M) &\approx&  \exp \left(Q(\bar{\theta}) \right) (2\pi)^{\frac{d}{2}} \det( H^{-1} ) )^{\frac{1}{2}},
\end{eqnarray}
where $d$ is number of free parameters in $\theta$.  If the conditions in  Proposition 3.4.3 in \cite{BD} hold,  we have for $n$ sufficiently large, $H \approx I(\bar{\theta}),$ where $I(\bar{\theta})$ is the Fisher information matrix.

Taking $2\log(\cdot)$ of both sides, (\ref{eqn:laplaceapx}) becomes
\begin{eqnarray*}
2\log(P(D|M)) &\approx&  2Q(\bar{\theta}) +d \log(2\pi) +   \log(\det( I(\bar{\theta})^{-1} ) )\\
 &=&  2\log(P(D|\bar{\theta}, M)) + 2 \log(P(\bar{\theta}| M)) +d \log(2\pi)  - \log(\det( I(\bar{\theta})).\\
\end{eqnarray*}

Now, we must calculate $-\log(\det( I(\bar{\theta}))$.  Define $n\equiv \sum_{i=1}^{C}n_i.$
Observe
\begin{eqnarray*}
I(\theta) &=&   \text{Var}\left[ \frac{\partial}{\partial \theta} \log(p(X,\theta)\right]\\
&=& \text{Var}\left[\sum_{i=1}^n \frac{\partial}{\partial \theta} \log(p(X_i,\theta))\right]\\
&=& \sum_{i=1}^C n_i I(\theta_i) \text{ by independence.}
\end{eqnarray*}
Henceforth, we will use $I_j$ to denote $I(\theta_j)$.  We will assume  that $I_1$ is positive definite, so that it can be written as $$I_1 = S_1^TS_1$$ for some non-singular matrix $S_1$.  Let $J$ denote the $d$-dimensional identity matrix.  For notational purposes, let $$n^* := \max_{2\leq j \leq C}( n_j ).$$

Observe 
\begin{eqnarray*}
\text{det}\left( \sum_{i=1}^C n_i I(\theta_i) \right) &=& \text{det}\left( S_1^T(n_1 J +\left(S_1^T\right)^{-1}\left(\sum_{i=2}^C n_i I(\theta_i) \right)S_1^{-1})S_1\right) \\
&=& \text{det}(I_1)\text{det}\left(n_1 J + n^*\left(S_1^T\right)^{-1}\left(\sum_{i=2}^C \frac{n_i}{n^*} I(\theta_i)\right)S_1^{-1}\right) \\
&=& \text{det}(I_1)n_1^d \text{det}\left( J + \frac{n^*}{n_1}\left(S_1^T\right)^{-1}\left( \sum_{i=2}^C \frac{n_i}{n^*} I(\theta_i)\right)S_1^{-1}\right) \\
&=& \text{det}(I_1)n_1^d \text{det}\left( J + \frac{n^*}{n_1}B\right), \\
\end{eqnarray*}
where $B =\left(S_1^T\right)^{-1})(\sum_{i=2}^C \frac{n_i}{n^*} I(\theta_i))S_1^{-1}$.  It should be noted that $B$ is a positive semi-definite matrix, as it is a $*-$transform of a sum of positive semi-definite matrices, so that Sylvester's Theorem states that $B$ has the same inertia as a positive semi-definite matrix.

Next, we would like to describe the growth of  $\text{det}\left( J + \frac{n^*}{n_1}B\right)$.
For any eigenvalue of $ J + \frac{n^*}{n_1}B$, say $\lambda$, we have by Weyl's theorem
$$1 \leq \lambda  \leq 1 + \frac{n^*}{n_1} \|B\|_2.$$
Observe
\begin{eqnarray*}
\|B\|_2 &=&\max_{\{x \in \mathbb{R}^d : \|x\|_2 = 1\}} x^T B x \\
&\leq&  \max_{\{x \in \mathbb{R}^d_1 : \|x\|_2 = 1\}} \|S_1^{-1}\|_2^2 x^T \left(\sum_{j=2}^C \frac{n_j}{n^*}I_j \right) x \\
 &\leq&  \max_{\{x \in \mathbb{R}^d : \|x\|_2 = 1\}} \|S_1^{-1}\|_2^2 \left(\sum_{j=2}^C \frac{n_j}{n^*}\|I_j\|_2^2 \right)\\
 &\leq&  \max_{\{x \in \mathbb{R}^d : \|x\|_2 = 1\}} \|S_1^{-1}\|_2^2  \left(\sum_{j=2}^C \|I_j\|_2^2 \right),
\end{eqnarray*}
which is independent of $n$.  Let $M_2 \equiv  \max_{\{x \in \mathbb{R}^d : \|x\|_2 = 1\}} \|S_1^{-1}\|_2^2  \left(\sum_{j=2}^C \|I_j\|_2^2 \right).$ 
Then,  we have
\begin{eqnarray*}
\text{det}( J + \frac{n^*}{n_1}B)  &=& \Pi_{m=1}^d \lambda_m \left( J + \frac{n^*}{n_1}B\right) \\
 &\leq& \left(1+\frac{n^*}{n_1}(M_2)\right)^d \\ 
  &\leq& \left(1+\frac{n - n_1}{n_1}(M_2)\right)^d \\
 \end{eqnarray*}
 
 The only term growing with $n$ above is $\frac{n - n_1}{n_1} ,$ the ratio of the supervised data over the unsupervised data.  In general, it is usually much more expensive to obtain additional supervised data than unsupervised;  thus, we find it reasonable to posit that $\frac{n - n_1}{n_1} \rightarrow 0$ in $n$ (i.e. is $o(1)$). In this case, $\log\left(\text{det}( J + \frac{n^*}{n_1}B)\right)$ is $o(1)$ by continuity and our bounds.
 
Hence, 
\begin{eqnarray}
\log\left(\text{det}\right(I(\theta)\left)\right) &=& \log\left(\text{det}(I_1)\right) + d \log(n_1) + \log\left( \text{det}( J + \frac{n^*}{n_1}B)\right)\\
&=& \log\left(\text{det}(I_1)\right) + d \log(n_1) + o(1).
\end{eqnarray}

When the posterior mode is nearly or is equal to the MLE, as is the case when the prior on $\theta$ is uniform and $\Theta$ is finite (cf. Bickel and Doksum pp. 114), substitute $\hat{\theta}$, the MLE, for $\bar{\theta}$, the posterior mode.
Then, we have
\begin{eqnarray*}
2\log(P(D|M)) &\approx&    2\log(P(D|\hat{\theta}, M)) + 2 \log(P(\hat{\theta}| M)) +d \log(2\pi)  - \log(\det( I_1(\bar{\theta}))\\
&=& 2\log(P(D|\hat{\theta}, M))  - d \log (n_1) + o\left(1\right) + O(1) \text{ using equation } (2).
\end{eqnarray*}
Note that the terms of order less than $O(1)$ get washed out in the limit as $n \rightarrow \infty$. If we drop them, we have our derivation of the adjusted BIC. 
\qed
                                                                                                                                                                                                                                              
\bibliographystyle{plainnat}

\end{document}